\documentclass[a4paper,UKenglish,autoref]{lipics-v2019}

\usepackage[ruled,vlined,linesnumbered,commentsnumbered]{algorithm2e}
\usepackage{xcolor}
\usepackage{hyperref}

\title{Extending the primal-dual 2-approximation algorithm beyond uncrossable set families}

\titlerunning{Extending the primal-dual 2-approximation algorithm beyond uncrossable set families}

\author{Zeev Nutov}{The Open University of Israel}{nutov@openu.ac.il}
{https://orcid.org/0000-0002-6629-3243}{}
\authorrunning{Zeev Nutov}
\Copyright{Zeev Nutov}

\nolinenumbers


\acknowledgements{}

\begin{document}

\maketitle
\newcommand {\ignore} [1] {}

\newcommand{\sem}    {\setminus}
\newcommand{\subs}   {\subseteq}
\newcommand{\empt}  {\emptyset}

\newcommand{\f}   {\frac}
\def\b   {\bar}

\def\A {\mathbb{A}}

\def\al   {\alpha}
\def\be {\beta}
\def\ga {\gamma}
\def\de   {\delta}
\def\eps {\epsilon}
\def\la {\lambda}

\def\AA  {{\cal A}}
\def\BB  {{\cal B}}
\def\CC  {{\cal C}}
\def\FF  {{\cal F}}
\def\LL  {{\cal L}}
\def\PP  {{\cal P}}

\def\sfec {{\sc Set Family Edge Cover}}

\def\cfo   {{\sc Constrained Forest}}
\def\sfo   {{\sc Steiner Forest}}
\def\st   {{\sc Steiner Tree}}
\def\rst   {{\sc Rooted Steiner Tree}}
\def\tj      {{\sc $T$-Join}}
\def\gpp {{\sc Generalized Point-To-Point Connection}}
\def\sna  {{\sc Steiner Network Augmentation}}

\def\tkcfo  {{\sc $(T,k)$-Constrained Forest}}        
\def\csf   {{\sc Constrained Steiner Forest}}

\keywords{primal dual algorithm, uncrossable set family, approximation algorithm}

\begin{abstract}
A set family $\FF$ is {\bf uncrossable} if $A \cap B,A \cup B \in \FF$ or $A \sem B,B \sem A \in \FF$ for any $A,B \in \FF$. 
A classic result of Williamson, Goemans, Mihail, and Vazirani [STOC 1993:708-717] 
states that the problem of covering an uncrossable set family by a min-cost edge set 
admits approximation ratio~$2$, by a primal-dual algorithm.
They asked whether this result extends to a larger class of set families and combinatorial optimization problems.
We define a new class of {\bf semi-uncrossable set families}, when for any $A,B \in \FF$ we have that 
$A \cap B \in \FF$ and one of $A \cup B,A \sem B ,B \sem A$ is in $\FF$, or $A \sem B,B \sem A \in \FF$.
We will show that the Williamson et al. algorithm extends to this new class of families
and  identify several ``non-uncrossable'' algorithmic problems that belong to this class.  
In particular, we will show that the union 
of an uncrossable family and a monotone family, or 
of an uncrossable family that has the disjointness property and a proper family,
is a semi-uncrossable family, that in general is not uncrossable.
For example, our result implies approximation ratio $2$ for the problem of finding a min-cost subgraph $H$ 
such that $H$ contains a Steiner forest and every connected component of $H$ 
contains at least $k$ nodes from a given set $T$ of terminals. 
\end{abstract}

\section{Introduction} \label{s:intro}

Let $G=(V,E)$ be  graph. 
For $J \subs E$ and $S \subs V$ let $\de_J(S)$ denote the set 
of edges in $J$ with one end in $S$ and the other in $V \sem S$, and let $d_J(S)=|\de_J(S)|$ be their number.
An edge set $J$ {\bf covers} $S$ if $d_J(S) \geq 1$.
Consider the following problem:

\begin{center}
\fbox{\begin{minipage}{0.98\textwidth} \noindent
\underline{\sfec} \\ 
{\em Input:} \ \ A graph $G=(V,E)$ with edge costs $\{c_e:e \in E\}$ and a set family $\FF$ on $V$. \\ 
{\em Output:} A min-cost edge set $J \subs E$ such that $d_J(S) \geq 1$ for all $S \in \FF$.
\end{minipage}} \end{center}

In this problem the set family $\FF$ may not given explicitly, but we will require that some queries related to $\FF$ can be 
answered in time polynomial in $n=|V|$. 
Specifically, following \cite{WGMV}, we will require that for any edge set $I$, 
the inclusion minimal members of the {\bf residual family  $\FF^I=\{S \in \FF:d_I(S)=0\}$} of $\FF$ 
can be computed in time polynomial in $n=|V|$.
We will also assume that $V,\empt \notin \FF$, as otherwise the problem has no feasible solution.

Various types of set families were considered in the literature, classified according 
to the conditions they satisfy. In this paper we will focus on the following four conditions:
\begin{tabbing}
{\bf Monotonicity}     \ \ \ \ \                       $\empt \neq S' \subs S$ \= $\in \FF$ \= $\Longrightarrow$ \= $S' \in \FF$ \\
{\bf Symmetry}         \hspace{2.35cm}                 $S$ \> $ \in \FF$ \> $\Longrightarrow$ \> $\b{S} \in \FF$, where $\b{S}=V \sem S$. \\
{\bf Disjointness} \hspace{0.685cm}  $\empt \neq S' \subs S$ \> $\in \FF$  \> $\Longrightarrow$ \> $S' \in \FF$ or $S \sem S' \in \FF$ \\
{\bf Uncrossability}  \hspace{1.3cm}             $A,B$ \>$ \in \FF$  \> $\Longrightarrow$ \> $A \cap B,A \cup B \in \FF$ or $A \sem B,B \sem A \in \FF$ 
\end{tabbing}

We note that in the paper of Goemans and Williamson \cite{GW}, the disjointness property was defined as follows:
if $A,B$ are disjoint and $A \cup B \in \FF$ then $A \in \FF$ or $B \in \FF$. 
For $A=S'$ and $B=S \sem S'$ this is equivalent to the definition given here.
Let us say that a set-family $\FF$ is:
\begin{itemize}
\item
{\bf Symmetric} if the Symmetry Condition holds for all $S \in \FF$.
\item
{\bf Monotone} if the Monotonicity Condition holds for all $\empt \neq S' \subs S \in \FF$.
\item
{\bf Disjointness compliable} if the Disjointness Condition holds for all $\empt \neq S' \subs S \in \FF$.
\item
{\bf Uncrossable} if the Uncrossability Condition holds for all $A,B \in \FF$.
\end{itemize}

A set family is {\bf proper} if it is both symmetric and disjointness compliable.
It is not hard to see the following:
\begin{itemize}
\item
If $\FF$ is monotone then $\FF$ is disjointness compliable and uncrossable. 
\item
If $\FF$ is proper then $\FF$ is uncrossabile.
\end{itemize}

Agrawal, Klein and Ravi \cite{AKR} designed and analyzed a primal-dual approximation algorithm
for the {\sc Steiner Forest} problem, and showed that it achieves approximation ratio $2$. 
A classic result of Goemans and Williamson \cite{GW} from the early 90's shows by an elegant proof 
that the same algorithm applies for {\sfec} with an arbitrary proper family $\FF$.
In fact, one of the main achievements of the Goemans and Williamson paper was defining 
a {\em generic class} of set families that models a rich collection of combinatorial optimization problems,
for which the primal dual algorithm achieves approximation ratio~$2$. 
Slightly later, Williamson, Goemans, Mihail, and Vazirani \cite{WGMV} 
further extended this result to the more general class of uncrossable families 
by adding to the algorithm a novel reverse-delete phase.
They posed an open question of extending this algorithm to 
a larger class of set families and combinatorial optimization problems.
However, for 30 years, the class of uncrossable set families remained the most general generic class 
of set families for which the primal dual algorithm achieves approximation ratio $2$.
We present a strict generalization of this result, illustrated by several examples. 

\begin{definition}
We say that a set family $\FF$ is {\bf semi-uncrossable} if for any $A,B \in \FF$ the following 
{\bf  semi-uncrossability condition} holds
$$
A \cap B \in \FF \mbox{ and one of } A \cup B,A \sem B ,B \sem A \mbox{ is in } \FF \ \ \mbox{ or } \ \ A \sem B,B \sem A \in \FF \ .
$$
\end{definition} 

One can see that if $\FF$ is symmetric and if
$A \cap B \in \FF$ or $A \sem B,B \sem A \in \FF$ whenever $A,B \in \FF$,  
then $\FF$ is uncrossable. 
Indeed, if $A,B \in \FF$ then $\b{A},\b{B} \in \FF$ (since $\FF$ is symmetric), and  
thus $\b{A} \cap \b{B} \in \FF$ or $\b{A} \sem \b{B},\b{B}\sem \b{A} \in \FF$; 
by symmetry we get that $A \cup B \in \FF$ or $B \sem A,A \sem B \in \FF$.
In particular, we get that:
$$
\mbox{\em If $\FF$ is symmetric and semi-uncrossable then $\FF$ is uncrossable.} 
$$
Namely, for symmetric $\FF$ the definitions of ``uncrossable'' and ``semi-uncrossable'' coincide.
However, for non-symmetric families, these definitions are distinct. For illustration,
consider the following example, see Fig.~\ref{f:eg}. 

\begin{figure} \centering \includegraphics{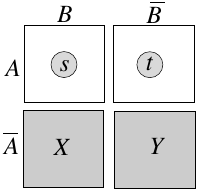}
\caption{An example of a semi-uncrossable family 
$\FF=\{A,B,\{s\},\{t\}\}$ that is not uncrossable.}
\label{f:eg} \end{figure}

\medskip

\noindent
{\bf Example 1.}
Let $V=X \cup Y \cup \{s,t\}$ where $X,Y \neq \empt$, let $A=\{s,t\}$ and $B=X \cup \{s\}$. Let 
$$
\FF=\{A,B, A \cap B,A \sem B\}=\{A,B,\{s\},\{t\}\} \ .
$$
This $\FF$ is semi-uncrossable but not uncrossable.
Note that an attempt to ``symmetrize'' $\FF$ by adding to $\FF$ the complements of the sets in $\FF$ 
results in a family that is {\em not} semi-uncrossable (and hence not uncrossable), where 
the semi-uncrossability condition fails for $\b{A},\b{B}$.
Thus semi-uncrossable families strictly generalize uncrossable set families. 
We prove the following.

\begin{theorem} \label{t:main}
{\sfec} with semi-uncrossable $\FF$ admits approximation ratio $2$.
\end{theorem}

The Theorem~\ref{t:main} algorithm is identical to the algorithm of Williamson, Goemans, Mihail, and Vazirani \cite{WGMV} for uncrossable families,
and the proof of the approximation ratio is also similar. 
Hence the contribution of our paper is not technical but rather conceptual, as follows.
\begin{itemize}
\item
We define a new class of set families that strictly contains the class of uncrossable families,
and show that the \cite{WGMV} algorithm achieves for this new class the same approximation ratio~$2$.
\item
We identify many combinatorial algorithmic problems that can be modeled by this class.  
\end{itemize}

Recently, Bansal, Cheriyan, Grout, and Ibrahimpur \cite{BCGI} defined a new 
class of set families which they called {\bf pliable set families with property $(\gamma)$};
a set family $\FF$ is {\bf pliable} if for any $A,B \in \FF$ at least two of the sets 
$A \cap B, A \cup B, A \sem B,B \sem A$ belongs to $\FF$, 
while {\bf property ($\gamma$)} imposes some restrictions on triples of sets from $\FF$.
They showed that the primal-dual algorithm of \cite{GW,WGMV} achieves approximation ratio $16$ for this class. 
Property $(\gamma)$ holds for semi-uncrossable families,
and the class of pliable set families with property $(\gamma)$ strictly includes the class of semi-uncrossable families. 
But handling this more general class of families comes with the price of a worse approximation ratio -- $16$ instead of $2$. 
The class of semi-uncrossable set families is sandwiched between 
the uncrossable families and the pliable families with property $(\gamma)$,
and our approximation ratio is $2$ -- the same as for uncrossable set families. 
Our paper is partly motivated by the result of \cite{BCGI}.
Another motivation comes from the {\sc Unbalanced Point-To-Point Connection} problem,
where the relevant set family $\FF$ has the property that 
$A \cap B \in \FF$ or $A \sem B,B \sem A \in \FF$; see \cite{HKKN}. 
This problem admits approximation ratio $O(\log n)$, 
but has no known super-constant approximation threshold. 

We illustrate application of Theorem~\ref{t:main} on {\bf combinations} of several fundamental problems,
among them  {\sfo}, {\tj}, {\sc Point-to-Point Connection}, $(T,k)$-{\cfo}, $k$-{\cfo}, and {\sna}; 
see formal definitions in the next section. 
We can combine any two of these problems by requiring that the solution will satisfy the requirements of both problems.
This means that the {\sfec} formulation of both problems has the same graph with the same edge costs, 
and we need to cover the union of the set families of the problems.\footnote{This 
definition of ``combination of problems'' depends on the specific set family that models the problem. 
For example, the set family of the {\sc Steiner Tree} problem is symmetric,
but the set family of the equivalent rooted version is not. 
For all problems considered, we will specify the set family we choose, 
trying to maximize the number of ``good'' properties, like uncrossability and symmetry.} 
Namely, if in one problem we need to cover a set family $\AA$ and in the other a set family $\BB$, then in 
the combination of the problems we need to cover the set family $\FF=\AA \cup \BB$.
In fact, in this way we can also combine a problem with itself. 
We will assume that each of $\AA,\BB$ is semi-uncrossable, 
as otherwise we should not expect that $\AA \cup \BB$ will be semi-uncrossable.
We will prove the following. 

\begin{theorem} \label{t:comb}
Let $\AA,\BB$ be semi-uncrossable set families over a groundset $V$. 
Then $\AA \cup \BB$ is semi-uncrossable in the following cases:
\begin{enumerate}[(i)]
\item
$\AA$ is monotone.
\item
$\AA$ is proper and $\BB$ is disjointness compliable.
\end{enumerate}
\end{theorem}

As was mentioned, for a symmetric family the definitions of ``uncrossable'' and ``semi-uncrossable'' coincide,
hence the new applications of Theorems \ref{t:main} and \ref{t:comb} are for non-symmetric families.
Consider for example the $k$-{\cfo} problem:  Given a graph $G=(V,E)$ with edge costs,
find a min-cost subgraph $H$ such that every connected component of $H$ has at least $k$ nodes. 
The set family we need to cover here is $\{S:1 \leq |S|<k\}$. This family is not symmetric, but it is monotone.  
Monotone set families were introduced by Goemans and Williamson \cite{GW-mon}, 
where they also showed several additional graph problems that can be modeled by such families.
Note that any monotone family is uncrossable, hence the problem of covering 
a monotone set family admits approximation ratio $2$.
Cou{\"{e}}toux, Davis, and Williamson \cite{CDW} improved the approximation ratio to $1.5$ by a dual fitting algorithm.

Clearly, we can achieve ratio $\al+\be$ for the problem of covering the union $\FF=\AA \cup \BB$ 
of an uncrossable family $\AA$ and a monotone family $\BB$, 
where $\al$ and $\be$ are the best known approximation ratios for covering $\AA$ and $\BB$, respectively. 
Currently, $\al=2$ \cite{WGMV} and $\be=1.5$ \cite{CDW}, so the overall approximation ratio will be $3.5$. 
Theorems \ref{t:main} and \ref{t:comb} give a better approximation ratio of $2$. 
Let us now illustrate Theorem~\ref{t:comb} and its limitations on few examples. 

\medskip

\noindent
{\bf Example 2.} 
{\em $\AA$ is monotone and $\BB$ is proper, $\AA \cup \BB$ is semi-uncrossable but not uncrossable.} \\
Consider the set family $\AA=\{S:1 \leq |S| \leq 2\}$ of the $3$-{\cfo} problem that is monotone, 
and the family $\BB=\{S:|S \cap \{s,t\}|=1\}$ of the {\sc Shortest $st$-Path} problem that is proper.
Now consider the sets $A=\{s,t\}$ and $B=X \cup \{s\}$ in Example~1, where $|X|,|Y| \geq 3$.
Then $A \in \AA$, $B \in \BB$, but among the sets $A \cap B, A \cup B, A \sem B, B \sem A$, 
only $A \cap B,A \sem B$ belong to $\AA \cup \BB$,
hence $\AA \cup \BB$ is not uncrossable. 
However, $\AA \cup \BB$ is semi-uncrossable by Theorem~\ref{t:comb}. 


\medskip

\noindent
{\bf Example 3.} {\em $\AA,\BB$ are both disjointness compliable, $\AA \cup \BB$ is not semi-uncrossable.} \\
Let $V=\{s,t,x,y\}$, $A=\{s,t\}$, $B=\{s,x\}$, and let 
$\AA=\{A, A \cap B\}=\{\{s,t\},\{s\}\}$ and $\BB=\{B, A \cap B\}=\{\{s,x\},\{s\}\}$.
Then $\AA,\BB$ are both uncrossable and disjointness compliable, $A \in \AA$, $B \in \BB$, 
but among the sets $A \cap B,A \cup B,A \sem B,B \sem A$ only $A \cap B$ belongs to $\AA \cup \BB$. 

\medskip

\noindent
{\bf Example 4.} 
{\em $\AA$ is proper and $\BB$ is uncrossable, $\AA \cup \BB$ is not semi-uncrossable.} \\ 
Let $V=\{x,y,s,t\}$, $A=\{x,s\}$, $B=\{x,y\}$, $\BB=\{B,\b{B}\}$ 
and let $\AA=\{S: |S \cap \{s,t\}|\}=1$ (this is the family of the {\sc Shortest $st$-Path} problem).
One can verify that $\AA$ is proper, $\BB$ is symmetric and uncrossable, $A \in \AA$, $B \in \BB$, 
but among the sets $A \cap B,A \cup B,A \sem B,B \sem A$ only $A \sem B,A \cup B$ belong to $\AA \cup \BB$,
so $\AA \cup \BB$ is not semi-uncrossable. 

\medskip

The proof of Theorem~\ref{t:comb} is by a standard case analysis and uncrossing. 
Yet, Theorem~\ref{t:comb} is sharp in the following sense.
\begin{itemize}
\item
In the combinations considered in the theorem, the combined set family $\FF=\AA \cup \BB$
may not be uncrossable even if $\AA,\BB$ are both uncrossable, as is shown in Example~2. 
\item
For any combination  not mentioned in Theorem~\ref{t:comb}, either 
the semi-uncrossability of $\AA \cup \BB$ can be deduced from the theorem,
or there exists an example such that $\AA \cup \BB$ is not semi-uncrossable
(namely, we cannot weaken the conditions in Theorem~\ref{t:comb}).
To see this,  note that in the following cases $\FF=\AA \cup \BB$ may not be semi-uncrossable.
\begin{itemize}
\item
$\AA,\BB$ are both disjointness compliable; see Example 3.
\item
If $\AA$ is proper and $\BB$ is symmetric; see Example 4.
\end{itemize}
\end{itemize}

We summarize this in Table~\ref{tbl:comb}.
Note that there is only entry in the table that guarantees an uncrossable family -- when both $\AA,\BB$ are proper,
while there are 5 entries where we get a semi-uncrossable family that may not be uncrossable.  
 
\begin{table} [htbp] 
\begin{center}
\begin{tabular}{|c|c|c|c|c|c|c|c|c|c|}  \hline  
                & d     & s    & m    & d,s    															
\\\hline 
d            & $-$  & $-$ & $+$ & $+$   
\\\hline 
s            &          & $-$ & $+$ & $-$    
\\\hline 
m         &          &         & $+$ & $+$ 
\\\hline
d,s        &         &         &         & \ $+^*$                          
\\\hline
\end{tabular}
\end{center}
\caption{Combinations that lead to semi-uncrossable families.
Here, ``d'',``s'', and ``m'' stand for  disjointness compliable, symmetric, and monotone, respectively. Also, 
``$+$'' means that the union is semi-uncrossable but is not uncrossable in general, 
``$+^*$'' means that the union is always uncrossable, 
while ``$-$'' means that the union is not semi-uncrossable in general (and there is an exampe for that).}
\label{tbl:comb}
\end{table}

In the next Section~\ref{s:comb} we prove Theorem~\ref{t:comb} 
and discuss some applications of Theorems \ref{t:main} and \ref{t:comb} for specific problems,
while Theorem~\ref{t:main} is proved in Section~\ref{s:main}. 

\section{Applications} \label{s:comb}

Recall that in Theorem~\ref{t:comb} we need to prove that if $\AA,\BB$ are semi-uncrossable set families then
$\AA \cup \BB$ is semi-uncrossable in the following cases:
\begin{enumerate}[(i)]
\item
$\AA$ is monotone.
\item
$\AA$ is proper and $\BB$ is disjointness compliable.
\end{enumerate}

Since each of $\AA,\BB$ is semi-uncrossable, we need to verify the semi-uncrossability condition 
only for pairs of sets $A,B$ such that $A \in \AA$ and $B \in \BB$.
Let $A \in A$ and $B \in B$.
Assume that all sets $A \cap B, A \sem B,B \sem A$ are non-empty,
as otherwise $\{A \cap B,A \cup B\}=\{A,B\}$ or $\{A \sem B,B \sem A\}=\{A,B\}$.
 
\medskip

Part (i) is trivial -- since $\AA$ is monotone, $A \cap B,A \sem B \in \AA$,
so the semi-uncrossability condition holds for $A,B$. 

\medskip

We prove (ii).  Since $\AA$ is symmetric, $\b{A} \in \AA$. 
Since $\AA,\BB$ are both disjointness compliable, the following holds:
\begin{enumerate}[(a)]
\item
$A \cap B \in \AA$ or $A \sem B \in \AA$.
\item
$\b{A} \cap B \in \AA$ or $\b{A} \sem B \in \AA$;  
note that $\b{A} \cap B=B \sem A$ and that $\b{A} \sem B=V \sem (A \cup B)$,
thus this is equivalent to $B \sem A \in \AA$ or $A \cup B \in \AA$ (by the symmetry of $\AA$). 
\item
$A \cap B \in \BB$ or $B \sem A \in \BB$.
\end{enumerate}
Let us consider the two cases in (c).
\begin{itemize}
\item
If $A \cap B \in \BB$, then the semi-uncrossability condition holds by (b).  
\item
If $B \sem A \in \BB$ then the semi-uncrossability condition holds by (a).
\end{itemize}
In both cases the semi-uncrossability condition holds, 
concluding the proof of (ii) and also of Theorem~\ref{t:comb}.

\medskip

We illustrate application of Theorems \ref{t:main} and \ref{t:comb} on combinations of several fundamental problems.
In all these problems we are given an (undirected) graph $G=(V,E)$ with edge costs $\{c_e:e \in E\}$,
and seek a subgraph of $G$ or an edge subset of $E$ that satisfies obeys a orescribed requirement. 
Below we indicate for each problem the additional parts of the input, the prescribed requirement,
the set family we need to cover, and the (strongest) relevant properties of this family.

\begin{center}
\fbox{\begin{minipage}{0.98\textwidth} \noindent
\underline{\sfo} ({\sc SF}) \\ 
Given a subpartition $\PP$ of $V$, find a min-cost subgraph $H$ of $G$ such that every part 
$P \in {\cal P}$ is contained in the same connected component of $H$. \\
{\em Set family:} $\{S: \empt \neq S \cap P \neq P \mbox{ for some } P \in \PP\}$; 
this family is proper and thus uncrossable.
\end{minipage}} \end{center}

\begin{center}
\fbox{\begin{minipage}{0.98\textwidth} \noindent
\underline{\tj} \\ 
Given a set $T \subs V$ of terminals with $|T|$ even,
find a min-cost subgraph $H$ of $G$ such that 
every connected component of $H$ contains an even number of terminals. \\
{\em Set family:} $\{S: |S \cap T| \mbox{ is odd}\}$; 
this family is proper and thus uncrossable.
\end{minipage}} \end{center}

\begin{center}
\fbox{\begin{minipage}{0.98\textwidth} \noindent
\underline{$(T,k)$-\cfo} ({\sc $(T,k)$-CF}) \\ 
Given a set $T \subs V$ of terminals and an integer $k$,
find a min-cost spanning subgraph $H$ of $G$ such that 
every connected component of $H$ has $0$ or at least $k$ terminals. \\
{\em Set family:} $\{S:1 \leq |S \cap T| <k\}$; 
this family is disjointness compliable and uncrossable. 
\end{minipage}} \end{center}

\begin{center}
\fbox{\begin{minipage}{0.98\textwidth} \noindent
\underline{$k$-{\cfo}} ({\sc $k$-CF}) \\
This is a particular case of  $(T,k)$-{\sc CF} when $T=V$. \\
{\em Set family:} $\{S:1 \leq |S| <k\}$; this family is monotone and thus uncrossable.
\end{minipage}} \end{center}

\begin{center}
\fbox{\begin{minipage}{0.98\textwidth} \noindent
\underline{\gpp} ({\sc G-P2P}) \\ 
Given integer charges $\{b(v):v \in V\}$ with $b(V)=0$,
find a min-cost spanning subgraph $H$ of $G$ such that 
every connected component of $H$ has zero charge. \\
{\em Set family:} $\{S:b(S) \neq 0\}$; 
this family is proper and thus uncrossable.
\end{minipage}} \end{center}

\begin{center}
\fbox{\begin{minipage}{0.98\textwidth} \noindent
\underline{\sna} ({\sc SNA}) \\ 
Given a graph $G_0=(V,E_0)$ with $E_0 \cap E=\empt$ and a set $D$ of demand node pairs, 
find a min-cost subgraph $H$ of $G$ such that 
$\la_{G_0 \cup H}(uv) \geq \la_{G_0}(u,v)+1$ for all $(u,v) \in D$, where
$\la_J(u,v)$ denotes the maximum number of edge disjoint $uv$-paths in a graph $J$.\\
{\em Set family:} $\{S: |S \cap \{u,v\}|=1, d_{G_0}(S)=\la_{G_0}(u,v) \mbox{ for some } \{u,v\} \in D\}$; 
this family is symmetric and uncrossable.
\end{minipage}} \end{center}

Combinations of problem pairs from the list above that 
give a semi-uncrossable family are summarized in Table~\ref{tbl:prob};
we added one additional row for the {\sc Shortest $st$-Path} problem,
that will help us to verify that the table was indeed filled correctly.   
Excluding the last row, we have overall 21 combinations, and 15 of them give a semi-uncrossable set family by Theorem~\ref{t:comb}.
Among these, 8 combinations give a set family that is not uncrossable in general. 
We now will justify whether each entry should be marked by 
``$+$'' (semi-uncrossable that may not be uncrossable), or  by ``$+^*$''  (always uncrossable),
or by ``$-$'' (may not be semi-uncrossable). 
Interestingly, Theorem \ref{t:comb} ``correctly filled'' the table -- 
the plus entries are exactly the ones that could be deduced from the theorem.

\begin{table} [htbp] 
\begin{center}
\begin{tabular}{|l|c|c|c|c|c|c|}  \hline  
                                                                                 & {\sc SF} {\bf d,s}  & {\tj} {\bf d,s}  & {\sc G-P2P} {\bf d,s} & {\sc $k$-CF} {\bf m}  & {\sc $(T,k)$-CF}	{\bf d} & {\sf SNA} {\bf s}																
\\\hline 
{\sc SF}                \hspace{1.31cm} {\bf d,s}        & $+^*$                   & $+^*$           & $+^*$                         & $+$                            & $+$                                 & $-$                                      
\\\hline 
{\tj}                         \hspace{0.7cm}   {\bf d,s}       &                              & $+^*$           &  $+^*$                         & $+$                           & $+$                                  & $-$
\\\hline 
{\sc G-P2P}          \hspace{0.7cm}  {\bf d,s}        &                              &                      & $+^*$                          & $+$                           & $+$                                  & $-$
\\\hline 
$k$-{\sc CF}        \hspace{0.55cm} {\bf m}      &                              &                      &                                     & \ $+^*$                      & $+$                                 & $+$
\\\hline 
$(T,k)$-{\sc CF}  \hspace{0.35cm}  {\bf d}         &                               &                     &                                     &                                    &  $-$                                 & $-$
\\\hline
{\sc SNA}             \hspace{1.32cm}    {\bf s}        &                              &                      &                                     &                                    &                                        & $-$
\\\hline \hline 
{\sc SP}                \hspace{1.28cm} {\bf d,s}   & $+^*$                   & $+^*$           & $+^*$                         & $+$                            & $+$                                 & $-$    
\\\hline 
\end{tabular}
\end{center}
\caption{Combinations that lead to semi-uncrossable families.
Here, ``d'',``s'', and ``m'' stand for  disjointness compliable, symmetric, and monotone, respectively 
(uncrossability is not mentioned since all problems are uncrossable). 
Also, ``$+$'' means that the union is semi-uncrossable, 
``$+^*$'' means that the union is uncrossable, 
while ``$-$'' means that the union is not semi-uncrossable in general.}
\label{tbl:prob}
\end{table}

We need the following observations to show that the table was indeed filled correctly. 
\begin{enumerate}[(A)]
\item
If $\FF$ is semi-uncrossable and symmetric then $\FF$ is uncrossable. 
\item
If $\AA,\BB$ are both monotone then $\AA \cup \BB$ is also monotone and thus uncrossable. 
\item
Except {\sc $k$-CF}, the {\sc Shortest $st$-Path} problem is a particular case of each of the problems considered here. 
\item
There is an example that the combination of {\sc Shortest $st$-Path} with {\sc $3$-CF} (and thus also with {\sc $(T,3)$-CF})
gives a semi-uncrossable family that is not uncrossable.
\item
There is an example that the combination of {\sc Shortest $st$-Path} with {\sc SNA} gives a set family that is not semi-uncrossable.
\end{enumerate}

Observation (A) was already proved in the introduction while Observation (B) is trivial.
For Observation (C),  one can verify that 
{\sc Shortest $st$-Path} is a particular case of: 
\begin{itemize}
\item {\sc SF} when $\PP=\{\{s,t\}\}$. 
\item  {\sc $T$-Join} when $T=\{s,t\}$.
\item {\sc G-P2P} when $b(s)=1,b(t)=-1$, and $b(v)=0$ otherwise. 
\item {\sc $(T,k)$-CF} when $T=\{s,t\}$ and $k=2$. 
\item {\sc SNA} when $E_0=\empt$ and $D=\{\{s,t\}\}$.
\end{itemize}

Observation (D) was proved in the Introduction, see Example~2. 

We now consider Observation (E). 
Note that in the combined problem of {\sc Shortest $st$-Path} with {\sc SNA}
we cannot just add the pair $\{s,t\}$ to the set $D$ of demand pairs, because 
we require that the augmenting graph $H$ will contain an $st$-path, 
and not the weaker condition that $G_0 \cup H$ will contain an $st$-path. 
Consider now the families $\AA,\BB$ in Example 4 in the Introduction. 
Let $E_0=\{xy,st\}$. Then (independently of the edge set $E$ and costs),
$\BB$ is the family of the {\sc SNA} instance and 
$\AA$ is the family of the {\sc Shortest $st$-Path} instance.
But in Example 4 the family $\AA \cup \BB$ is not semi-uncrossable. 

\medskip

We use these observations to justify the last row of the table -- of the {\sc Shortest $st$-Path} problem. 
The first 3 entries in the row are by Theorem~\ref{t:comb}(ii). 
The next two entries are by Theorem~\ref{t:comb}(i) and Observation~(D).
The last entry in the row is by Observation (E). 

In a similar way we can fill the other entries.
Note that observation (D) is used only for the plus entries, to decide 
whether the entry should be ``$+$'' or ``$+^*$''.

\medskip

\noindent
{\bf The first 3 columns:}
Here all combinations give a semi-uncrossable family by Theorem~\ref{t:comb}(ii).
The family is also uncrossable, since the combined set family is symmetric. 

\medskip

\noindent
{\bf The {\sc $k$-CF} column.}
The first 3 entries in the column follow from Theorem~\ref{t:comb}(i) and Observations (C) and (D).
The fourth entry is by Observation (B).
\medskip

\noindent
{\bf The {\sc $(T,k)$-CF} column.}
All positive entries in the column of {$(T,k)$-CF} are by Theorem~\ref{t:comb} and Observations (C) and (D);
in the first 3 entries of the column, instead of Observation (C) and (D) we can use 
the fact that  {\sc $k$-CF} is a particular case of {\sc $(T,k)$-CF}.

\medskip

\noindent
{\bf The {\sc SNA} column.}
All the minus entries in this column follow from Observations (C) and (E). 
The unique plus entry in this column is by Theorem~\ref{t:comb}(i) and Observation (D). 

\newpage

\section{Proof of Theorem~\ref{t:main}} \label{s:main}

In this section we prove Theorem~\ref{t:main}.
During the proof we indicate the parts that extend or fail to pliable set families;
recall that a set family $\FF$ is pliable if for any $A,B \in \FF$ at least $2$ sets
from $A \cap B,A \cup B,A \sem B,B \sem A$ belong to $\FF$. 
Also recall that the residual family of $\FF$ w.r.t. an edge set $I$ is $\FF^I=\{S \in \FF:d_I(S)=0\}$. 

One can see that if an edge $e$ covers one of the sets $A \cap B, A \cup B, A\sem B,B \sem A$ 
then it also covers one of $A,B$. This implies the following.

\begin{corollary} \label{c:res}
If $\FF$ is semi-uncrossable or if $\FF$ is pliable then so is $\FF^I$, for any edge set $I$.
\end{corollary}

An {\bf $\FF$-core} is an inclusion minimal member of $\FF$; let $\CC(\FF)$ denote the family of $\FF$-cores. 
Similarly to uncrossable families, we have the following.

\begin{lemma} \label{l:CX}
If $\FF$ is semi-uncrossable then for any $C \in \CC(\FF)$ and $S \in \FF$,
either $C \subs S$ or $C \cap S=\empt$; in particular, the $\FF$-cores are pairwise disjoint. 
\end{lemma}

Note that if $\FF$ is pliable then the cores are also pairwise disjoint.
However, the first property in Lemma~\ref{l:CX} may not hold for pliable families.

We now describe the algorithm. 
Consider the following LP-relaxation {\bf (P)} for {\sfec} and its dual program {\bf (D)}:
\[ \displaystyle
\begin{array} {lllllll} 
&  \hphantom{\bf (P)} & \min         & \ \displaystyle \sum_{e \in E} c_e x_e & 
   \hphantom{\bf (P)} & \max         & \ \displaystyle \sum_{S \in \FF} y_S  \\
&      \mbox{\bf (P)} & \ \mbox{s.t.}  & \displaystyle \sum_{e \in \delta(S)} x_e  \geq 1 \ \ \ \ \ \forall S \in \FF \ \ \ \ \ \ \  &
       \mbox{\bf (D)} & \ \mbox{s.t.}  & \displaystyle \sum_{\delta(S) \ni e} y_S \leq c_e \ \ \ \ \ \forall e \in E \\
&  \hphantom{\bf (P)} &              & \ \ x_e \geq 0 \ \ \ \ \ \ \ \ \ \ \ \forall e \in E &
   \hphantom{\bf (P)} &              & \ \ y_S \geq 0 \ \ \ \ \ \ \ \ \ \ \ \ \forall S \in \FF. 
\end{array}
\]

Given a solution $y$ to {\bf (D)}, an edge $e \in E$ is {\bf tight}
if the inequality of $e$ in {\bf (D)} holds with equality.
The algorithm has two phases.

{\bf Phase~1} starts with $I=\emptyset$ an applies a sequence of iterations.
At the beginning of an iteration, we compute the family $\CC(\FF^I)$.
Then we raise the dual variables corresponding to the members of $\CC(\FF^I)$
uniformly (possibly by zero), until some edge $e \in E \sem I$ becomes tight, and add $e$ to $I$.
Phase I terminates when $\CC(\FF^I)=\empt$, namely when $I$ covers $\FF$.

{\bf Phase~2} applies on $I$ ``reverse delete'', which means the following.
Let $I=\{e_1, \ldots, e_j\}$, where $e_{i+1}$ was added after $e_i$. 
For $i=j$ downto $1$, we delete $e_i$ from $I$ if $I \sem \{e_i\}$ still covers $\FF$.
At the end of the algorithm, $I$ is output.

It is easy to see that the produced dual solution is feasible, hence 
$\sum_{S \in \FF} y_S \leq {\sf opt}$, by the Weak Duality Theorem.
We prove that at the end of the algorithm 
$$
\sum_{e \in I} c(e) \leq 2\sum_{S \in \FF} y_S \ .
$$
As any edge in $I$ is tight, the last inequality is equivalent to
$$
\sum_{e \in I} \sum_{\de_l(S) \ni e} y_S \leq 2 \sum_{S \in \FF} y_S \ .
$$
By changing the order of summation we get:
$$
\sum_{S \in \FF} d_I(S) y_S \leq 2 \sum_{S \in \FF} y_S \ .
$$
It is sufficient to prove that at any iteration the increase at the left hand side is at most
the increase in the right hand side. 
Let us fix some iteration, and let $\CC$ be the family of cores among 
the members of $\FF$ not yet covered. The increase in the left hand side is 
$\varepsilon \cdot \sum_{C \in \CC} d_I(C)$, 
where $\varepsilon$ is the amount by which the dual variables were raised in the iteration, 
while the increase in the right hand side is 
$\varepsilon \cdot 2 |\CC|$. Consequently, it is sufficient to prove that
$\sum_{C \in \CC} d_I(C) \leq 2 |\CC|$.
As the edges were deleted in reverse order, the set $I'$ of edges in $I$ that were 
added after the iteration (and ``survived'' the reverse delete phase),
form an inclusion minimal edge-cover of the family $\FF'$ of members 
in $\FF$ that are uncovered at the beginning of the iteration.
Note also that $\bigcup_{C \in \CC} \delta_{I}(C) \subseteq I'$.
Hence to prove ratio $2$, it is sufficient to prove the following purely combinatorial statement, 
in which due to Corollary~\ref{c:res} we can revise our notation to $\FF \gets \FF'$ and $I \gets I'$.

\begin{lemma} \label{l:local-ratio}
Let $I$ be an inclusion minimal cover of a set family $\FF$ and let $\CC=\CC(\FF)$ be the family of $\FF$-cores. 
If $\FF$ is semi-uncrossable then
\begin{equation} \label{e:CC}
\sum_{C \in \CC} d_I(C) \leq 2 |\CC|-1 \ .
\end{equation}
\end{lemma}

In the rest of this section we prove Lemma~\ref{l:local-ratio}.
Let us say that {\bf two sets $A,B$ are laminar} if they are disjoint or one of them contains the other;
namely, if at least one of the sets $A \cap B, A \sem B, B \sem A$ is empty. 
A set family $\LL$ is a {\bf laminar set family} if its members are pairwise laminar. 
Let $I$ be an inclusion minimal edge cover of a set family $\FF$. 
We say that $S_e \in \FF$ is a {\bf witness set} for $e$ if 
$e$ is the unique edge in $I$ that covers $S_e$, namely, if $\de_I(S_e)=\{e\}$.
We say that $\LL \subs \FF$ is a {\bf witness family} for $I$ if 
$|\LL|=|I|$ and for every $e \in I$ there is a witness set $S_e \in \LL$.
By the minimality of $I$, there exists a witness family $\LL \subs \FF$.
We will show the following.

\begin{lemma} \label{l:witness}
Let $I$ be an inclusion minimal cover of a set family $\FF$. 
If $\FF$ is semi-uncrossable then there exists a witness family $\LL \subs \FF$ for $I$ that is laminar.
\end{lemma}
\begin{proof}
Let $A,B \in \FF$ be witness sets of edges $e,f \in I$, respectively. 
Note that no edge in $I \sem \{e,f\}$ covers a set from $A \cap B,A \cup B,A \sem B,B \sem A$,
as such an edge covers one of $A,B$, contradicting that $A,B$ are witness sets.   
Thus all possible locations of such $e,f$ are as depicted in Fig.~\ref{f:uncross}.
We claim that one of the sets $A \cap B, A \sem B, B \sem A$ is a witness set for one of $e,f$. 
This follows from the following observations.
\begin{itemize}
\item
$A \sem B \notin \FF$ in (a) and $B \sem A \notin \FF$ in (b); in both cases, 
$A \cap B \in \FF$ is a witness set for one of $e,f$.
\item
$A \cup B \notin \FF$ in (c), and $A \cap B \notin \FF$ in (d); in both cases, 
$A \sem B,B \sem A \in \FF$ and each of $A \sem B,B \sem A$ is a witness set for one of $e,f$.
\end{itemize}

\ignore{---------------------
\begin{enumerate}[(a)]
\item
In case (a) 
$A \sem B \notin \FF$; hence $A \cap B \in \FF$ is a witness set for $e$.
\item
In case (b) 
$B \sem A \notin \FF$; hence $A \cap B \in \FF$ is a witness set for $f$.
\item
In case (c) ($e \subs B$ and $f \subs A$),  $A \cup B \notin \FF$; 
hence $A \sem B \in \FF$ is a witness set for $f$ or $B \sem A \in \FF$ is a witness set for $e$.
\item
In case (d) ($e \subs V \sem B$ and $f \subs V \sem A$),  $A \cap B \notin \FF$; 
hence $A \sem B \in \FF$ is a witness set for $e$ and $B \sem A \in \FF$ is a witness set for $f$.
\end{enumerate}
--------------------}

By the minimality of $I$ there exists a witness family for $I$. 
Let $\LL$ be a witness family for $I$ with $\sum_{S \in \LL}|S|$ minimal. 
We claim that $\LL$ is laminar. 
Suppose to the contrary that there are $A,B \in \LL$ that are not laminar.
Then there is $A' \in \{A \cap B,A \sem B,B \sem A\}$ and $B' \in \{A,B\}$ such that
$\LL'=(\LL \sem \{A,B\})\cup \{A',B'\}$ is also a witness family for $I$.
However, $\sum_{S \in \LL'} |S|<\sum_{S \in \LL}|S|$, contradicting the choice of $\LL$.
\end{proof}

\begin{figure} \centering \includegraphics{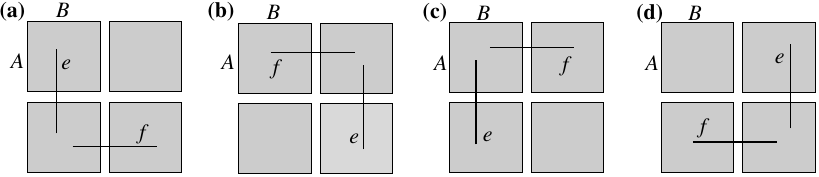}
\caption{Illustration to the proof of Lemma~\ref{l:witness}.}
\label{f:uncross} \end{figure}

We note that Lemma~\ref{l:witness} extends to pliable families 
by a slightly more complicated proof, see \cite[Lemma~10]{BCGI}. 

Let $J=\bigcup_{C \in \CC} \de_I(C)$. Note that $\de_I(C)=\de_J(C)$ for any $C \in \CC$.
As any subfamily of a laminar family is also laminar, we conclude from Lemma~\ref{l:witness}
that if $\FF$ is semi-uncrossable then there exists a laminar witness family $\LL \subs \FF$ for $J$. 
Note that for every $S \in \LL$ there is $C \in \CC$ such that $C \subseteq S$.
Moreover, if $\FF$ is semi-uncrossable then by Lemma~\ref{l:CX} $\LL \cup \CC$ is laminar 
(this is not true if $\FF$ is only pliable).
Thus to finish the proof of Lemma~\ref{l:local-ratio}, it is sufficient to prove the following lemma,
that was essentially proved in \cite{GW,WGMV}. 

\begin{lemma} \label{l:2C}
Let $\LL,\CC$ be set families such that 
$\LL \cup \CC$ is laminar, 
the members of $\CC$ are pairwise disjoint, and $\CC$ is the family of inclusion minimal members of $\LL \cup \CC$. 
Let $J$ be an edge set such that every edge in $J$ covers some $C \in \CC$. 
If $\LL$ is a witness family for $J$ then 
\begin{equation} \label{e:2C}
\sum_{C \in \CC} d_J(C) \leq 2 |\CC|-1 \ .
\end{equation}
\end{lemma}
\begin{proof}
Let $S=S_e$ be an inclusion minimal set in $\LL$, let $C \in \CC$ such that $C \subs S_e$ (possibly $C=S_e$),
and let $\LL_C$ be the family of sets in $\LL$ that contain $C$ and contain no other set in $\CC$. We claim that then:
\begin{enumerate}[(i)]
\item
$\de_J(C)=\{e\}$.
\item
$S_e \in \LL_C$ and $|\LL_C| \leq 2$.
\end{enumerate}

We prove (i). Suppose to the contrary that there is $f \in \de_J(C) \sem \{e\}$. 
Since $\LL$ is laminar and $S_e$ is inclusion minimal in $\LL$, we must have 
$S_e \subset S_f$ or $S_e \cap S_f=\empt$. In both cases $f$ covers $S_e$,
contradicting that $S_e$ is a witness set for $e$. 

We prove (ii). It is easy to see that $S_e \in \LL_C$. 
Suppose to the contrary that there are distinct $S_f,S_g \in \LL_C \sem \{S_e\}$.
Then $S_e \subset S_f \subset S_g$. 
Since every edge in $J$ covers some $C \in \CC$, there is $C_f \in \CC$ that is covered by $f$. 
Since $S_g \in \LL_C$, we must have $C_f \cap S_g=\empt$.
But then $f$ must cover $S_g$, contradicting that $S_g$ is a witness set for $g$. 

The rest of the proof is by induction on $|\CC|$. 
The base case $|\CC|=1$ is easy to verify. 
We prove the lemma for $|\CC| \geq 2$, assuming that it holds for $|\CC|-1$. 

Let $C$ be as in (i) and let $\CC'=\CC \sem \{C\}$. Define $J'$ and $\LL'$ as follows.
\begin{itemize}
\item
If $\LL_C=\{S_e\}$ and if there is $C' \in \CC'$ such that $\de_J(C')=\{e\}$,
then $J'=J$ and $\LL'=(L \sem \{S_e\}) \cup \{C'\}$.
\item
Else, $J'=J \sem \{e\}$ and $\LL'=\LL\sem \{S_e\}$.
\end{itemize}
One can verify that the triple $\LL',\CC',J'$ satisfies the assumptions of the lemma and that
$\sum_{C \in \CC}d_J(C)-\sum_{C \in \CC'}d_{J'}(C) \leq 2$.
From this, using the induction hypothesis we get
$$
\sum_{C \in \CC}d_J(C) \leq \sum_{C \in \CC'}d_{J'}(C)+2 \leq (2|\CC'|-1)-1+2=2|\CC|-1 \ ,
$$
concluding the proof.
\end{proof}

This concludes the proof of Theorem~\ref{t:main}. 


\section{Concluding remarks} \label{s:concl}

In this paper we presented a new class of semi-uncrossable set families 
that extends the class of uncrossable families. 
The algorithm of \cite{WGMV} for uncrossable families 
extends to the more general class of semi-uncrossable families with the same approximation ratio of $2$.
Our work shows that often a combination of two problems such that at least one of them is non-symmetric gives 
a problem with a semi-uncrossable family that is not uncrossable in general. 

Jain \cite{J} introduced the iterative rounding method, 
and used it to obtain approximation ratio $2$ for the problem of covering a 
skew-supermodular (a.k.a. weakly supermodular) set function;
the class of  skew-supermodular set functions generalizes the class of uncrossable set families. 
It is an interesting question whether Jain's result can be extended to a larger class of set functions.


\begin{thebibliography}{1}

\bibitem{AKR}
A.~Agrawal, P.~Klein, and R.~Ravi.
\newblock When trees collide: An approximation algorithm for the generalized
  {Steiner} problem on networks.
\newblock {\em SIAM J. on Computing}, 24(3):440--456, 1995.

\bibitem{BCGI}
I.~Bansal, J.~Cheriyan, L.~Grout, and S.~Ibrahimpur.
\newblock Improved approximation algorithms by generalizing the primal-dual
  method beyond uncrossable functions.
\newblock {\em CoRR}, abs/2209.11209, 2022.
\newblock URL: \url{https://arxiv.org/abs/2209.11209}.

\bibitem{CDW}
B.~Cou{\"{e}}toux, J.~M. Davis, and D.~P. Williamson.
\newblock A 3/2-approximation algorithm for some minimum-cost graph problems.
\newblock {\em Mathematical Programming}, 150:19--34, 2015.

\bibitem{GW-mon}
M.~X. Goemans and D.~P. Williamson.
\newblock Approximating minimum-cost graph problems with spanning tree edges.
\newblock {\em Operations Research Letters}, 16:183--189, 1994.

\bibitem{GW}
M.~X. Goemans and D.~P. Williamson.
\newblock A general approximation technique for constrained forest problems.
\newblock {\em SIAM J. Comput.}, 24(2):296--317, 1995.

\bibitem{HKKN}
MohammadTaghi Hajiaghayi, Rohit Khandekar, Guy Kortsarz, and Zeev Nutov.
\newblock On fixed cost k-flow problems.
\newblock {\em Theory Comput. Syst.}, 58(1):4--18, 2016.

\bibitem{J}
K.~Jain.
\newblock A factor 2 approximation algorithm for the generalized steiner
  network problem.
\newblock {\em Combinatorica}, 21(1):39--60, 2001.

\bibitem{WGMV}
D.~P. Williamson, M.~X. Goemans, M.~Mihail, and V.~V. Vazirani.
\newblock A primal-dual approximation algorithm for generalized steiner network
  problems.
\newblock {\em Combinatorica}, 15(3):435--454, 1995.

\end{thebibliography}

\end{document}